\documentclass[11pt]{article}
\usepackage[T2A]{fontenc}
\usepackage[cp1251]{inputenc}
\usepackage[english]{babel}
\usepackage{amssymb}
\usepackage{amsmath}
\usepackage{amsthm}
\usepackage{amsfonts}
\usepackage{amssymb}

\newtheorem{theorem}{Theorem}

\newtheorem{proposition}{Proposition}

\newtheorem{definition}{Definition}

\newtheorem{example}{Example}
\newtheorem{remark}{Remark}

\begin{document}
\begin{center}
\textbf{\uppercase{L\'evy Differential Operators and Gauge Invariant Equations for Dirac and Higgs Fields}}\\
Boris O. Volkov\\

Steklov Mathematical Institute of Russian Academy of Sciences,\\ ul. Gubkina 8, Moscow, 119991 Russia\\E-mail: borisvolkov1986@gmail.com

\end{center}

Dedicated to the memory of Sergei Starodubov
\\

{\bf Abstract:} 
We study the L\'evy infinite-dimensional differential operators (differential operators defined  by the analogy with the L\'evy Laplacian) and their relationship to the Yang-Mills equations. We consider the parallel transport on the space of curves as an infinite-dimensional analogue of chiral fields and show  that it is a solution to the system of differential equations  if and only if the associated connection is a solution to the Yang-Mills equations.
This system is an analogue of the equation of  motion of chiral fields and contains the L\'evy divergence. The systems of infinite-dimensional equations containing  L\'evy differential operators, that are equivalent to the Yang-Mills-Higgs equations and the Yang-Mills-Dirac equations (the equations of quantum chromodynamics), are obtained.  
The equivalence of two ways to define L\'evy differential operators is shown.

keywords: L\'evy Laplacian,  Yang-Mills equations, Yang-Mills-Higgs equations, QCD equations


%46N20 

 \section*{Introduction}
 
The current paper has  two goals. We provide a classification of infinite-dimensional differential operators defined by the analogy with the L\'evy Laplacian and naturally connected to the Yang-Mills equations. And we describe by such differential operators  solutions to the  
Yang-Mills-Higgs equations and  the QCD equations (the equations of quantum chromodynamics).

In the current paper  we will use the following scheme of the definition of differential operators, and not the form in which they were originally introduced.
 Let $E$  be a real normed space  and  $S$  be  a linear  functional acting on  a subspace of  $L(E,E^\ast)$.\footnote{If $X$ and $Y$ are normed spaces, the symbol $L(X,Y)$ denotes the space of linear continuous operators from $X$ to $Y$. The symbol  $C^k(X,Y)$ denotes the space of $k$-times Fr\'echet differentiable functions  from $X$ to $Y$.}  Then $S$ defines a linear differential operator  of the second order $D^{2,S}$
  acting on   $f\in C^2(E,\mathbb{R})$
by the formula
$$
D^{2,S}f(x)=S(f''(x)),
$$
and the  linear differential  operator of the first order  $D^{1,S}$ acting on
$B\in C^1(E,E^\ast)$  by the formula
$$
D^{1,S}B(x)=S(B'(x)).
$$
 For example, if we choose $E=\mathbb{R}^d$ and  $S=tr$ (trace), then
   $D^{2,tr}$ is the Laplacian  $\Delta$ and  
    $D^{1,tr}$ is the divergence $div$. 
Now let $E$ be a real normed space continuously embedded in a real separable Hilbert space  $H$. Let $\{e_n\}$ be an orthonormal basis in  $H$ that consists of elements of  $E$.
 Then the value of the L\'evy trace $tr_L^{\{e_n\}}$ (generated by the orthonormal basis $\{e_n\}$)  on %$K\in L(L_2(0,1),L_2(0,1))$ 
$K\in L(E,E^\ast)$ is defined by
 \begin{equation}
  \label{eq2}
 tr^{\{e_n\}}_L K=\lim_{n\to\infty}\frac 1n \sum_{k=1}^n <Ke_k, e_k>.
 \end{equation}
The L\'evy Laplacian (generated by the orthonormal basis $\{e_n\}$) $\Delta_L^{\{e_n\}}$ is the differential operator  $D^{2,tr^{\{e_n\}}_L}$.
This operator was introduced by Paul L\'evy for the case $E=H=L_2(0,1)$ (see Ref.~\cite{L1951}).

Another original definition of the L\'evy Laplacian from Ref.~\cite{L1951} (the definition of the L\'evy Laplacian by  means of a second-order derivative of special form) will be given in terms of the L\'evy trace  as follows.
Let  $K$ be a bilinear functional on $L_2(0,1)$ such that for all $u,v\in L_2(0,1)$ the following equality holds
 \begin{equation}
 \label{1}
 K(u,v)=\int_0^1\int_0^1 K_V(t,s) u(t)v(s)dtds+\int_0^1 K_L(t)u(t)v(t)dt,
 \end{equation}
where $K_V\in L_2([0,1]\times[0,1])$ and $K_L\in L_\infty[0,1]$. 
The value of the  L\'evy trace   $tr_L$ on $K$ is defined by
$$
tr_LK=\int_0^1 K_L(t)dt.
$$
Then the L\'evy Laplacian  $\Delta_L$ is the differential operator $D^{2,tr_L}$. 
The orthonormal bases $\{e_n\}$ in $L_2(0,1)$, such that   $tr_L$ is a restriction of $tr^{\{e_n\}}_L$,
form a special class of weakly uniformly dense  bases. Hence for weakly uniformly dense  bases
 $\Delta_L$ is a restriction of   $\Delta^{\{e_n\}}_L$.

The article deals with functionals which are generalizations of the traces described above.  Using uniformly dense bases it is proved that these functionals  coincide under certain conditions.
Such functionals will be also called the L\'evy traces. 
 We consider differential operators of the first and the second order (Laplacian, d'Alembertian, divergence) generated by these traces  arising in the study of gauge fields. Such operators will be called the L\'evy differential operators.

The following papers are devoted to the connection between the L\'evy  differential operators and the Yang-Mills 
equations. In Ref.~\cite{AV1981} by Aref'ieva and Volovich the functional divergence was introduced by the analogy with the L\'evy Laplacian defined  by  means of a second-order derivative of special form. The  relationship of this divergence with gauge fields were studied. In the current paper it is shown that the functional divergence is a restriction of the L\'evy divergence.
In Refs.~\cite{AGV1993,AGV1994} Accardi, Gibilisco and Volovich introduced an analogue of the L\'evy Laplacian acting on functions on the space of piecewise smooth functions (see also Ref.~\cite{Accardi}). This operator was defined by a more complicated  form of the second derivative than in~(\ref{1}). 
 It is also called the L\'evy Laplacian.  In Refs.~\cite{AGV1993,AGV1994} the following  was proved. The connection in the trivial bundle over a Euclidean space is a solution to the Yang-Mills equations if and only if the parallel transport, considered as an operator-valued functional on a space of curves, is a solution to  the Laplace equation for such L\'evy Laplacian.
Moreover  in Refs.~\cite{AGV1993,AGV1994} the L\'evy d'Alembertian was introduced. The similar theorem  about the equivalence of the Yang-Mills equations for a connection on a Minkowski space and the L\'evy-d'Alembert equation for the parallel transport holds for this operator. The case of  a Riemannian
manifold  and the stochastic case were considered by Leandre and Volovich in Ref.~\cite{LV2001}.
In Refs.~\cite{Volkov2012,Volkovdiss} the theorem about the connection between the L\'evy Laplacian and the gauge fields was proved in the case of a manifold  and the L\'evy Laplacian and the L\'evy  d'Alembertian defined as the C\'esaro mean of the second directional derivatives (see also Ref.~\cite{AS2006} for the definition). Thus, the problem has been  solved, set in Ref.~\cite{AGV1993}. In Refs.~\cite{VolkovD,VolkovLLI} it was proved that the  L\'evy Laplacians and L\'evy d'Alembertians defined as the  Ces\`aro mean of the second directional derivatives and defined by a special form of the second derivative coincide in the plane case.  The instantons were described in terms of  the L\'evy Laplacian and the parallel transport in Ref.~\cite{VolkovLLI}. In Refs.~\cite{Volkov2017,Volkov2018} the Levy Laplacian defined as the Chesaro mean of the directional derivatives in the stochastic case was studied.\footnote{Note that the different approach to the Yang-Mills equations based on the stochastic parallel transport but without using the L\'evy differential operators was considered in Refs.~\cite{Bauer1998,ABT,AT}.} It was shown that, unlike the deterministic case, the equivalence of the Yang-Mills equations and the Levy-Laplace equation is not valid for such Laplacian. It means that this Levy Laplacian does not coincide with the Levy Laplacian introduced   in Ref.~\cite{LV2001}.
 
  The path 2-form is a function  on the space of curves in $\mathbb{R}^d$ beginning at the origin such  that the value of this function on a curve  is a matrix-valued 2-form at the endpoint of the curve. 
In  Ref.~\cite{Gross} by Gross   the infinite-dimensional non-commutative Poincar\'e lemma was proved. It implies that an infinite-dimensional 1-form  (a smooth function from the space of curves to the space of linear matrix-valued operators on the space of curves) is generated by the path 2-form and  is closed if and only if  it is generated by a parallel transport associated with a  connection in the trivial bundle over $\mathbb{R}^d$. Using this non-commutative Poincar\'e lemma in Ref.~\cite{Gross} Gross  proved that the differential equations on the path 2-form, that are generalization of Maxwell's equations, are  equivalent to the Yang-Mills equations.
 Also in  Ref.~\cite{Gross}  the differential equations on the path 2-forms and the infinite-dimensional fields,  that are equivalent to the QCD equations and  the Yang-Mills-Higgs equations,  were obtained.

 In the current  paper we show  that     an infinite-dimensional  1-form generated by the path 2-form  is closed and the value of the  L\'evy divergence on this 1-form is zero if and only if it is associated with a finite-dimensional connection, that is a solution to the Yang-Mills equations.\footnote{The converse statement was proved in Ref.~\cite {AV1981} for the functional divergence.} Thus, we obtain the differential equations that are analogue of the equations of motion of the chiral field and are equivalent to the Yang-Mills equations (the parallel transport can be considered as an analogue  of the chiral field).
 In addition, we  study the connection between the  L\'evy infinite-dimensional differential operators and Yang-Mills-Higgs equations and the  QCD equations. For this purpose  the equations derived  in Ref.~\cite{Gross} are  used.
 But in contrast to our approach the L\'evy differential operators were not  considered in Ref.~\cite{Gross}. 
  We  generalize the Accardi-Gibilisco-Volovich theorem and reformulate the Yang-Mills-Higgs equations and the  QCD equations in terms of the parallel transport and the L\'evy d'Alembertian. Using the non-commutative Poincar\'e lemma  we  reformulate
  the Yang-Mills-Higgs equations and the  QCD equations in terms of the infinite-dimensional 1-form and the L\'evy divergence.
   The relationship of the L\'evy d'Alembertian  and the QCD equations was considered in Ref.~\cite{VolkovD}.

The paper is organized as follows. The first section provides the definition of the L\'evy  traces as the Ces\`aro mean of diagonal elements and differential operators  generated by these traces. The second section provides the definition of the L\'evy traces as an integral functional and associated  differential operators.
Also in the second section the equivalence of two ways to define L\'evy differential operators is shown. The connection between the L\'evy differential operators and the Yang-Mills equations is discussed in the third section. In the fourth and the fifth sections  systems of infinite-dimensional differential equations, containing  L\'evy differential  operators, which are equivalent to the Yang-Mills-Higgs equations and to the QCD equations  respectively, are obtained.

\section{L\'evy trace and L\'evy Differential  Operators}
\label{sec1}

Let us recall  the general scheme of the definition of  homogeneous linear differential
operators  from the  paper Ref.~\cite{ASF} (see also Ref.~\cite{S}).
Let $X,Y,Z$ be real normed vector spaces.  Let
$C^n(X,Y)$ be the space of
 $n$ times Frech\'et differentiable functions from $X$  to $Y$.
Then for any $x\in X$ it is hold  that $f^{(n)}(x)\in L_n(X,Y)$,
where the space $L_n(X,Y)$ is defined by induction:
$L_1(X,Y)=L(X,Y)$ and $L_n(X,Y)=L(X,L_{n-1}(X,Y))$.
Let $S$ be a linear operator from $dom S$ to $Z$, where $dom S\subset L_n(X,Y)$.
The domain $dom D^{n,S}$ of the differential operator  $D^{n,S}$ of order $n$  generated by the linear operator  $S$ consists of all $f\in C^n(X,Y)$ such that for all $x\in E$  holds
  $f^{(n)}(x)\in dom S$.
Then the value of $D^{n,S}$ on $f\in dom D^{n,S}$ is defined by the following formula
 $$
 D^{n,S} f(x)=S(f^{(n)}(x)).
 $$
If we choose $X=\mathbb{R}^d$, $Y=Z=\mathbb{R}$ and $S=tr$, then
 $D^{2,tr}$ is the Laplace operator $\Delta$. If we choose $X=Y=\mathbb{R}^d$,
  $Z=\mathbb{R}$ and $S=tr$, then 
 $D^{1,tr}$ is the divergence $div$ acting on the space of $C^1$-smooth 1-forms.

Let $E_1$ be a real normed space. 
 Let $E_1$ be continuously embedded  in a real separable  infinite-dimensional Hilbert space   $H_1$ in such  a way that the image of  $E_1$ is dense in $H_1$.
 Set $E=\mathbb{R}^d\otimes E_1$ and $H=\mathbb{R}^d\otimes H_1$.\footnote{We assume that the topology of the tensor product of $\mathbb{R}^d\otimes E_1$  is defined by some cross-norm (see Ref.~\cite{sch}). We assume that $H=\mathbb{R}^d\otimes H_1$ is a  Hilbert tensor product.}
Then $E\subset H\subset E^\ast$ is a rigged Hilbert space. Let  $\{e_n\}$ be an orthonormal basis in  $H_1$ that consists of elements of
$E_1$.  Let $\{p_1,\ldots,p_d\}$ be an orthonormal basis in  $\mathbb{R}^d$.
Denote  the Euclidean metric on  $\mathbb{R}^d$ by $\delta=(\delta_{\mu\nu})$. Also let the Minkowski metric  $\eta=(\eta_{\nu\mu})$ be
defined on $\mathbb{R}^d$.
 We assume  that this metric    is diagonal $\{\eta_{\nu\mu}\}=diag\{1,-1,\ldots,-1\}$ in the basis $\{p_1,\ldots,p_d\}$. Everywhere below $g\in\{\delta,\eta\}$.

 \begin{definition}
The L\'evy trace $tr^{\{e_n\},g}_{L}$ generated by the orthonormal basis $\{e_n\}$  and the metric $g$  is a linear functional on $dom\,tr^{\{e_n\},g}_{L}$ defined by
 \begin{multline}
 \label{tlL1}
 tr^{\{e_n\},g}_{L}(T)=\\=\lim_{n\to\infty}\frac 1{n}\sum_{k=1}^n(<T(p_1\otimes e_k),p_1\otimes e_k>+\sum_{\mu=2}^d\sigma_g^2<T(p_\mu\otimes e_k),p_\mu\otimes e_k>),
 \end{multline}
 where $\sigma_g=1$ or $\sigma_g=i$, if  $g=\delta$ and $g=\eta$ respectively, and $dom\,tr^{\{e_n\},g}_{L}$ consists of all $T\in L(E,E^\ast)$ for which the right-hand side of~(\ref{tlL1}) exists.  
 \end{definition}

\begin{definition}
Choose $X=E$ and $Y=Z=\mathbb{R}$. The L\'evy Laplacian  $\Delta^{\{e_n\}}_L$ generated by the orthonormal basis $\{e_n\}$ is the differential operator $D^{2,tr_L^{\{e_n\},\delta}}$. The L\'evy d'Alembertian $\Box^{\{e_n\}}_L$  generated by the orthonormal basis  $\{e_n\}$  is the differential operator $D^{2,tr_L^{\{e_n\},\eta}}$.
\end{definition}

 \begin{definition}
Choose $X=E$, $Z=\mathbb{R}$  and $Y=E^\ast$.
In the Euclidean case the L\'evy divergence $div^\delta_L$ generated by the orthonormal basis $\{e_n\}$ is the operator $D^{1,tr_L^{\{e_n\},\delta}}$. In the Minkowski case the L\'evy divergence $div^\eta_L$ generated by the orthonormal basis $\{e_n\}$ is the operator $D^{1,tr_L^{\{e_n\},\eta}}$.
\end{definition}

\begin{remark} 
The L\'evy trace was  introduced for  $d=1$ in Ref.~\cite{AS1993}. 
\end{remark}
\begin{remark}
The non-classical L\'evy Laplacian is a generalization of the L\'evy Laplacian defined 
in the following way. 
Let $R$ be a linear mapping from\newline$span\{e_n\colon n\in \mathbb{N}\}$ to $E_1$. 
The non-classical L\'evy trace $tr^{\{e_n\},g}_{LR}$ generated by the orthonormal basis  $\{e_n\}$,  the metric  $g$  and the linear operator $R$ acts on $T\in L(E,E^\ast)$ by the formula
 \begin{multline}
 \label{tlLLL1}
 tr^{\{e_n\},g}_{LR}(T)=\lim_{n\to\infty}\frac 1{n}\sum_{k=1}^n(<T(p_1\otimes Re_k),p_1\otimes Re_k>+\\+\sum_{\mu=2}^d\sigma_g^2<T(p_\mu\otimes Re_k),p_\mu\otimes Re_k>),
 \end{multline}
 where $\sigma_g=1$ or $\sigma_g=i$, if $g=\delta$ or $g=\eta$ respectively. The non-classical L\'evy Laplacian  $\Delta^{\{e_n\}}_{LR}$ and the non-classical L\'evy d'Alembertian  $\Box^{\{e_n\}}_{LR}$ are the differential operators  $D^{2,tr_{LR}^{\{e_n\},\delta}}$ and $D^{2,tr_{LR}^{\{e_n\},\eta}}$ respectively.  The non-classical L\'evy divergence  $div^g_{LR}$  is the differential operator $D^{1,tr^g_{LR}}$.
The non-classical  L\'evy trace and the associated non-classical L\'evy Laplacian were introduced for $d=1$ in Ref.~\cite{AccSm2007}. As we show below, such operators can be useful in the study of gauge fields (see also~\cite{Volkov2018}). 
Moreover, the exotic L\'evy Laplacians (see Refs.~\cite{AS1993,AS2009} for the definition), which are studied  in the white noise analysis  (see Ref.~\cite{AJS2015} and references therein),  can be represented as   non-classical L\'evy Laplacians (see Refs.~\cite{Volkov2013,VolkovD}).
\end{remark}

\section{L\'evy trace as an integral functional}
\label{sec2}

Let   $$W^{1,p}_0([0,1],\mathbb{R}^{d})=\{\text{$\sigma$ is absolutely continuous, $\sigma(0)=0$, $\dot{\sigma}\in L_p((0,1),\mathbb{R}^d)$}\}$$ 
and $C_0^1([0,1],\mathbb{R}^d)=\{\sigma\in C^1([0,1],\mathbb{R}^d) ,\sigma(0)=0\}$.

Below, unless specifically stated,
$$E\in\{C_0^1([0,1],\mathbb{R}^d),W^{1,1}_0([0,1],\mathbb{R}^{d}),
W^{1,2}_0([0,1],\mathbb{R}^{d})\}.$$
Denote by $E_0$ the space of all $\sigma\in E$ such that $\sigma(1)=0$.

Denote by $T^2_E$ the space of continuous  bilinear functionals on $E\times E$, which restriction on   $E_0\times E_0$ have the following form
\begin{multline}
\label{f''}
Q(u,v)=\int_0^1\int_0^1Q^V_{\mu\nu}(t,s)u^{\mu}(t)v^{\nu}(s)dtds+\\
+\int_0^1Q^L_{\mu\nu}(t)u^{\mu}(t)v^{\nu}(t)dt
%+v^{\mu}(r)u^{\nu}(r))dr+\\
+\frac 12\int_0^1Q^S_{\mu\nu}(t)(\dot{u}^{\mu}(t)v^{\nu}(t)+\dot{v}^{\mu}(t)u^{\nu}(t))dt,%\, u,v\in E_0
\end{multline}
where
$Q^V_{\mu\nu}\in L_1([0,1]\times[0,1],\mathbb{R})$,
$Q^L_{\mu\nu}\in L_1([0,1],\mathbb{R})$,
$Q^S_{\mu\nu}\in L_\infty([0,1],\mathbb{R})$, $Q^L_{\mu\nu}$  is a symmetric tensor:
$Q^L_{\mu\nu}=Q^L_{\nu\mu}$,
and $Q^S_{\mu\nu}$ is an anti-symmetric tensor:
$
Q^S_{\mu\nu}=-Q^S_{\nu\mu}$.
\begin{definition}
The L\'evy trace $tr^g_L$ generated by the metric $g$ is a linear functional on $T^2_E$ defined by the formula
$$
tr^g_L Q=\int_0^1Q^L_{\mu\nu}(t)g^{\mu\nu}dt.
$$
\end{definition}

\begin{remark}
The first term in~(\ref{f''}) is the Volterra part, the second term is the L\'evy part, the third term is a singular part.
\end{remark}
\begin{definition}
Choose $X=E$, $Z=Y=\mathbb{R}$.
The L\'evy Laplacian  $\Delta_L$ is the differential operator $D^2_{tr^\delta_L}$. The  L\'evy d'Alembertian 
$\Box_L$ is the differential operator $D^2_{tr^\eta_L}$.
\end{definition}

\begin{definition}
Choose $X=E$, $Z=\mathbb{R}$ and $Y=E^\ast$.
In the Euclidean case the L\'evy divergence $div^\delta_L$  is the operator $D^1_{tr^\delta_L}$. In the Minkowski case the L\'evy divergence $div^\eta_L$ is the differential operator $D^1_{tr^\eta_L}$.
\end{definition}

The definition of the functional divergence by Aref'ieva and Volovich (see Ref.~\cite{AV1981})  can be reformulated as follows. Let $E=C^1([0,1],\mathbb{R}^d)$ and let  $T^2_{v}$
be the space of $Q\in T^2_E$ such that $Q_{\mu\mu}\in C([0,1]\times [0,1])$ for any $\mu\in\{1,\ldots,d\}$   and
$\sum_{\mu=1}^d\int_0^1Q^V_{\mu\mu}(t,t)dt=0$.
Let $tr^\delta_{v}$ be the restriction of $tr^\delta_{L}$ on the space $T^2_{v}$. The functional divergence is the differential operator $D^{1,tr^\delta_{v}}$. Thus, the functional divergence is a restriction of the L\'evy divergence.

Recall the following definition from Ref.~\cite{L1951} (see also Refs.~\cite{F2005,KOS}).

\begin{definition} An orthonormal basis $\{e_n\}$ in $L_2(0,1)$
is weakly uniformly dense (or equally dense) if 
\begin{equation}
\label{unifdence}
\lim_{n\to\infty}\int_0^1 h(t)(\frac 1n\sum_{k=1}^n
e_k^2(t)-1)dt=0
\end{equation}
for any $h\in L_{\infty}[0,1]$.
\end{definition}

\begin{example}
The sequence $e_n(t)=\sqrt{2}\sin{n\pi t}$ is a weakly uniformly dense basis  in 
 $L_2(0,1)$.
\end{example}

\begin{theorem}
\label{p}
Let  $\{e_n\}$ be a weakly uniformly dense  orthonormal basis in $L_2(0,1)$ that consists of functions from $E$ such that
$e_n(1)=0$ for all $n\in\mathbb{N}$. Let  the sequence $\{e_n\}$ be   uniformly bounded. 
If $Q\in T^2_E$, 
then
\begin{equation}
\label{tr=tr}
tr^{g,\{e_n\}}_L(Q)=tr_L^{g}(Q)
\end{equation}
\end{theorem}
\begin{proof}
%Set $l_n(t,s)=\frac 1n\sum_{k=1}^n e_k(s)e_k(t)$.
Since $e_n(0)=e_n(1)=0$ for all $n\in\mathbb{N}$ and $Q^S_{\mu\nu}$ is anti-symmetric tensor, 
the following equality holds
\begin{multline*}
\frac 1n\sum_{k=1}^n  Q(p_\mu e_k,p_\mu e_k)=\frac 1n\sum_{k=1}^n\int_0^1\int_0^1Q^V_{\mu\mu}(t,s)e_k(t)e_k(s)dtds+\\+\int_0^1Q^L_{\mu\mu}(t)(\frac 1n\sum_{k=1}^n
e_k^2(t))dt.
\end{multline*}
For a weakly uniformly dense and uniformly bounded  orthonormal basis in $L_2(0,1)$ equality~(\ref{unifdence}) holds for any $h\in L_1[0,1]$ and also
$$
\lim_{n\to\infty}\int_0^1\int_0^1K(s,t)e_n(s)e_n(t)dtds=0
$$
holds for any $K\in L_1([0,1]\times[0,1])$ (see Ref.~\cite{KOS}).
Since $Q^L_{\mu\mu}\in L_1(0,1)$ and $Q^V_{\mu\mu}\in L_1([0,1]\times [0,1])$ we obtain~(\ref{tr=tr}).

\end{proof}

The following theorem is a direct corollary of  Theorem~\ref{p}.
\begin{theorem}
Let  an orthonormal basis  $\{e_n\}$ in $L_2(0,1)$ satisfy the conditions of Theorem~\ref{p}.
If $f\in dom \Delta_L$, then 
 $\Delta^{\{e_n\}}_L f=\Delta_L f$ 
  and $\Box^{\{e_n\}}_L f=\Box_L f$.
If  $B\in dom\,div^g_L$, then 
 $div^{g,\{e_n\}}_L B=div^g_L B$.
\end{theorem}
\begin{remark}
The definitions of the L\'evy Laplacian and the L\'evy d'Alembertian from  section~\ref{sec1} and section~\ref{sec2} can be naturally transferred to the space\newline$C^2(E,M_m(\mathbb{C}))$. Additionally  the definitions
of the L\'evy traces can be transferred to the space $L(E,E_0^\ast)$.
This fact allows  to define the L\'evy divergence on the space $C^1(E,L(E_0,M_m(\mathbb{C})))$.
Such divergence will be used in the following sections.
\end{remark}

\begin{remark}
\label{nonclassical}
Choose $E=H=W^{1,2}_0([0,1],\mathbb{R}^{d})\}$.
Let  $\{f_n\}$ be an orthonormal basis in $E_1=W^{1,2}_0([0,1],\mathbb{R})\}$ defined by
$f_1(t)=t$ and $f_n(t)=\frac {\sqrt2} {\pi (n-1)}\sin(\pi (n-1) t)$ for $n>1$. Consider the operator $N\colon span\{f_n\colon n\in \mathbb{N}\} \to E_1$ defined by
$Nf_n=(n-1)f_n$. Then $tr^{\{f_n\},g}_{LN}$ coincides with  $tr^{g}_{L}$ on the space $T^2_E$.
The differential operators generated by the trace $tr^{g}_{L}$  are restrictions  of the corresponding
differential operators generated by the trace $tr^{\{f_n\},g}_{LN}$. 

\end{remark}

\section{The Yang-Mills equations}

The summation over repeated upper and lower indices is assumed.
Let $G$ be a closed Lie group realized as a subgroup of $U(N)$, and let $Lie(G)$  be a Lie algebra realized as a subalgebra of $u(N)$.
A connection in the trivial vector bundle with base  $\mathbb{R}^{d}$, fiber $\mathbb{C}^N$, structure group $G$ is defined as a  $Lie(G)$-valued $C^\infty$-smooth 1-form $A_\mu(x)dx^\mu$ on $\mathbb{R}^{d}$. For the function $\varphi\in C^1(\mathbb{R}^{d},Lie(G))$ the  covariant derivative along the field $(\frac{\partial}{\partial_\mu})$ is defined by
 $$\nabla_\mu \varphi=\partial_\mu \varphi+[A_\mu,\varphi].$$
 The tensor of curvature is defined by $Lie(G)$-valued  2-form $$F(x)=\sum_{\mu<\nu} F_{\mu \nu}(x)dx^\mu\wedge dx^\nu,$$ where $F_{\mu\nu}=\partial_\mu A_\nu-\partial_\nu A_\mu+[A_\mu,A_\nu]$.
 The connection  $A$ is  a solution to the Yang-Mills equations with the current $j(x)=j_\mu(x)dx^\mu$ if
\begin{equation}
\label{YM}
g^{\lambda \mu}\nabla_\lambda F_{\mu \nu}=j_\nu.
\end{equation}

 For any curve $\sigma\in E$  an operator $U^A_{t,s}(\sigma)$, where $0\leq s\leq t \leq1$, is defined as a  solution to the system of differential equations
 \begin{equation*}
\label{form222}
 \left\{
\begin{aligned}
\frac
{d}{dt}U^A_{t,s}(\sigma)=-A_{\mu}(\sigma(t))\dot{\sigma}^{\mu}(t) U^A_{t,s}(\sigma)\\
\frac
d{ds}U^A_{t,s}(\sigma)=U^A_{t,s}(\sigma)A_{\mu}(\sigma(s))\dot{\sigma}^{\mu}(s)\\
\left.U^A_{t,s}(\sigma)\right|_{t=s}=I_N.
\end{aligned}
\right.
\end{equation*}
Here and below the symbol $I_N$ denotes the unitary matrix in $M_N(\mathbb{C})$.
The symbol $U^A_{s,t}(\sigma)$ denotes $U^A_{t,s}(\sigma)^{-1}$,  where $0\leq s\leq t \leq1$.
 The operator $U^A_{1,0}(\sigma)$ is a parallel transport along $\sigma$.

\begin{proposition}
It is valid that
$U^A_{1,0}\in C^\infty(E,M_N(\mathbb{C}))$.
If $u,v\in E$, then
\begin{equation*}
\label{u'}
(U_{1,0}^A)'(\sigma)(u)=-\int_0^1U^A_{1,t}(\sigma)F_{\mu\nu}(\sigma(t))u^\mu(t)\dot{\sigma}^\nu(t)U^A_{t,0}(\sigma)dt
%-\\
%-A_\mu(\sigma(t))u^\mu(t)U_{t,s}(\sigma)-
-A_\mu(\sigma(1))u^\mu(1)U^A_{1,0}(\sigma).
\end{equation*}
\end{proposition}
For a proof of this proposition we refer the reader to Ref.~\cite{Gross} and also to Ref.~\cite{Driver}.

The results of Proposition~\ref{AGVpr1}, of Theorem~\ref{AGV1} and of Theorem~\ref{LLYM} for $j=0$ 
belong to Accardi, Gibilisco and Volovich. 

\begin{proposition}
\label{AGVpr1}
If $\sigma\in E$ and $u,v\in E_0$,
the following equality holds
\begin{multline}
\label{u''}
(U_{1,0}^A)''(\sigma)(u,v)
=\int_0^1\int_0^1K^V_{\mu\nu}(\sigma)(t,s)u^{\mu}(t)v^{\nu}(s)dtds+\\
+\int_0^1K^L_{\mu\nu}(\sigma)(t)u^{\mu}(t)v^{\nu}(t)dt
%+v^{\mu}(r)u^{\nu}(r))dr+\\
+\frac 12\int_0^1K^S_{\mu\nu}(\sigma)(t)(\dot{u}^{\mu}(t)v^{\nu}(t)+\dot{v}^{\mu}(t)u^{\nu}(t))dt,
\end{multline}
where
\begin{equation}
\label{KV}
K^V_{\mu\nu}(\sigma)(t,s)=\begin{cases}
U^A_{1,t}(\sigma)F_{\mu\lambda}(\sigma(t))\dot{\sigma}^\lambda(t)U^A_{t,s}(\sigma)F_{\nu\kappa}(\sigma(s))
\dot{\sigma}^\kappa(s)U^A_{s,0}(\sigma),&\text{if $t\geq s$}\\
U^A_{1,s}(\sigma)F_{\nu\kappa}(\sigma(s))\dot{\sigma}^\kappa(s)U^A_{s,t}(\sigma)F_{\mu\lambda}(\sigma(t))\dot{\sigma}^\lambda(t)U^A_{t,0}(\sigma)
,&\text{if $t<s$},
\end{cases}
%U_{1,r}(\sigma)(-F_{\mu\nu}(\sigma(r))\dot{\sigma}^\nu(r)U_{r,s}(\sigma)(-F_{\mu\nu}(\sigma(s))\dot{\sigma}^\nu(s)U_{s,0}(\sigma)\theta(r-s)+\\
%+\theta(s-r),
\end{equation}
$$
K^L_{\mu\nu}(\sigma)(t)=\frac 12U^A_{1,t}(\sigma)(-\nabla_\mu F_{\nu\lambda}(\sigma(t))\dot{\sigma}^\lambda(t)-\nabla_\nu F_{\mu\lambda}(\sigma(t))\dot{\sigma}^\lambda(t))U^A_{t,0}(\sigma),
$$
$$
K^S_{\mu\nu}(\sigma)(t)=U^A_{1,t}(\sigma)F_{\mu\nu}(\sigma(t))U^A_{t,0}(\sigma).
$$
\end{proposition}
For a proof of this proposition we refer the reader to Ref.~\cite{AGV1994} and also to Ref.~\cite{VolkovLLI}.

The following theorem  is a direct corollary  of Proposition~\ref{AGVpr1}.
\begin{theorem}
\label{AGV1}
The following equality holds
$$
D^2_{tr^g_L}U_{1,0}^A(\sigma)=-\int_0^1 U^A_{1,t}(\sigma)g^{\mu\nu}\nabla_\mu F_{\nu\lambda}(\sigma(t))\dot{\sigma}^\lambda(t)U^A_{t,0}(\sigma)dt
$$
\end{theorem}

 If $\sigma\in E$, everywhere below the symbol $\sigma^r$, where $0\leq r \leq 1$, denotes the curve from $E$ defined by
 $\sigma^r(t)=\sigma(rt)$.

\begin{theorem}
\label{LLYM}
The following two assertions are equivalent:
\begin{enumerate}
\item the connection  $A$ is a solution to the Yang-Mills equations with the current~(\ref{YM});
\item    the equality 
\begin{equation}
\label{2assertion}
D^2_{tr^g_L}U_{1,0}^A(\sigma)=-\int_0^1 U^A_{1,t}(\sigma)j_\nu(\sigma(t))\dot{\sigma}^\nu(t)U^A_{t,0}(\sigma)dt.
\end{equation}
holds for the parallel transport  $U^A_{1,0}$.
\end{enumerate}
\end{theorem}

\begin{proof}
Let~(\ref{2assertion}) hold. Note that ${\dot{\sigma}^{r\nu}(t)}=r{\dot{\sigma}^{\nu}(rt)}$ and $U^A_{(rt),(rs)}(\sigma)=U^A_{t,s}(\sigma^r)$. Let $\sigma\in E\cap C^1([0,1],\mathbb{R}^d)$. Then 
\begin{multline}
\label{vicl}
\int_0^r U^A_{1,t}(\sigma)(-g^{\lambda\mu}\nabla_\lambda F_{\mu\nu}(\sigma(t))\dot{\sigma}^\nu(t)U^A_{t,0}(\sigma)dt=\\=U^A_{1,r}(\sigma)D^2_{tr^g_L}U_{1,0}^A(\sigma^r)=-U^A_{1,r}(\sigma)\int_0^1 U^A_{1,t}(\sigma^r)j_\nu(\sigma^r(t))\dot{\sigma^r}^\nu(t)U^A_{t,0}(\sigma^r)dt=\\
=\int_0^r U^A_{1,t}(\sigma)(-g^{\lambda\mu}\nabla_\lambda F_{\mu\nu}(\sigma(t))\dot{\sigma}^\nu(t)U^A_{t,0}(\sigma)dt.
\end{multline}
Differentiating~(\ref{vicl}) with respect to $r$, we  obtain
$$
U^{A}_{1,r}(\sigma)g^{\lambda\mu}\nabla_\lambda F_{\mu\nu}(\sigma(r))\dot{\sigma}^\nu(r)U^A_{r,0}(\sigma)=
U^{A}_{1,r}(\sigma)j_\nu(\sigma^r)\dot{\sigma}^\nu(r)U^A_{r,0}(\sigma).
$$
Choose an  appropriate $\sigma\in E\cap C^1([0,1],\mathbb{R}^d)$, we see that the connection $A$ is a solution to the Yang-Mills equations with current~(\ref{YM}).
The other side of the theorem is obvious.
\end{proof}

Consider the function $B^A\colon  E\to L(E_0,Lie(G))$ defined by
$$B^A(\sigma)u=U^A_{0,1}(\sigma)\partial_uU^A_{1,0}(\sigma),\ \sigma\in E,\,u\in E_0.$$
It is known that  $B^A\in C^\infty(E,L(E_0,Lie(G)))$ (see Ref.~\cite{Gross}).

\begin{proposition}
\label{divL}
The following equalities hold
\begin{multline}
\label{laplas}
div_L^gB^A(\sigma)=-\int_0^1U^A_{0,t}(\sigma)g^{\lambda \mu}\nabla_\lambda F_{\mu\nu}(\sigma(t))
\dot{\sigma}^\nu(t)U^A_{t,0}(\sigma)dt=
\\=U_{0,1}^A(\sigma)D^2_{tr^g_L}U_{1,0}^A(\sigma).
\end{multline}
\end{proposition}
\begin{proof}
If $u,v\in E_0$, the following equality holds
\begin{equation}
\partial_u B^A(\sigma)v=\partial_v U^A_{0,1}(\sigma)\partial_u U^A_{1,0}(\sigma)+U^A_{0,1}(\sigma)\partial_v \partial_u U^A_{1,0}(\sigma).
\end{equation}
Then, using~(\ref{u''}), we obtain that
for all $\sigma \in E$ and $u,v\in E_0$
the following holds
\begin{multline}
\partial_{u}B^A(\sigma)v=\int_0^1\int_0^1R^V_{\mu\nu}(t,s)u^{\mu}(t)v^{\nu}(s)dtds+\\
+\int_0^1R^L_{\mu\nu}(t)u^{\mu}(t)v^{\nu}(t)dt
%+v^{\mu}(r)u^{\nu}(r))dr+\\
+\frac 12\int_0^1R^S_{\mu\nu}(t)(\dot{u}^{\mu}(t)v^{\nu}(t)+\dot{v}^{\mu}(t)u^{\nu}(t))dt,
\end{multline}
where
\begin{multline*}
\label{RV}
R^V_{\mu\nu}(\sigma)(t,s)=\\=
\begin{cases}[U^A_{0,s}(\sigma)F_{\nu\kappa}(\sigma(s))\dot{\sigma}^\kappa(s)U^A_{s,0}(\sigma),U^A_{0,t}(\sigma)F_{\mu\lambda}(\sigma(t))\dot{\sigma}^\lambda(t)U^A_{t,0}(\sigma)]
,\text{if $t\leq s$}\\
0,\text{if $t>s$},
\end{cases}
\end{multline*}
$$
R^L_{\mu\nu}(\sigma)(t)=\frac 12U^A_{0,t}(\sigma)(-\nabla_\mu F_{\nu\lambda}(\sigma(t))\dot{\sigma}^\lambda(t)-\nabla_\nu F_{\mu\lambda}(\sigma(t))\dot{\sigma}^\lambda(t))U^A_{t,0}(\sigma),
$$
$$
R^S_{\mu\nu}(\sigma)(t)=U^A_{0,t}(\sigma)F_{\mu\nu}(\sigma(t))U^A_{t,0}(\sigma).
$$
Thus, we obtain~(\ref{laplas}).
\end{proof}

The following result belongs to Gross (see Ref.~\cite{Gross}). It is an analogue of the Poincar\'e lemma.

\begin{theorem}
\label{GrossThm}
Let $E\in\{W^{1,1}_0([0,1],\mathbb{R}^{d}),
W^{1,2}_0([0,1],\mathbb{R}^{d})\}$.
Let a function $B$ from $E$ to $L(E_0, Lie(G)))$ (an infinite-dimensional 1-form) is defined by
\begin{equation}
\label{Bx}
B(\sigma)u=\int_0^1h(\sigma^r)<\dot{\sigma}(r),u(r)>dr,
\end{equation}
where $\sigma\in E$, $u\in E_0$, $h(\sigma)$ is a $Lie(G)$-valued anti-symmetric tensor for any $\sigma\in E$, the function $E\ni \sigma\mapsto h(\sigma)$ is continuously Frech\'et differentiable.\footnote{The function $h$ is called the path 2-form.} 
Let 1-form $B$ be closed in the sense that for all $\sigma\in E$ and $u,v\in E_0$
the following equality holds
$$
\partial_uB(\sigma)v-\partial_vB(\sigma)u+[B(\sigma)u,B(\sigma)v]=0%, u,v\in E_0
$$
Then  there is a  connection $A$ such that $B=B^A$.
\end{theorem}
In the Gross's theorem the space $\mathcal{P}$ of piecewise  $C^\infty$-smooth functions with the norm $\|\sigma\|=\int_0^1|\dot{\sigma}(t)|dt$  instead  of $E$ was considered. The proof 
 can be transferred unchanged to the case of the spaces  $W^{1,1}_0([0,1],\mathbb{R}^{d})$ and $W^{1,2}_{0}([0,1],\mathbb{R}^d)$.

\begin{remark}
It is known (see Ref.~\cite{Gross}) that if an infinite-dimensional 1-form  $B$ can be represented as~(\ref{Bx}) then $B\in C^\infty(E,L(E_0,Lie(G))$. It was also  proved in Ref.~\cite{Gross}  that for two connection  $A$ and $A'$ such that  $B^A=B^{A'}$ there exists a unique function $b\in C^\infty(\mathbb{R}^d,G)$ with $b(0)=I_N$ such that $A(x)=b^{-1}(x)A'(x)b(x)+b^{-1}(x)db(x)$. 
\end{remark}

\begin{theorem}
\label{thmdiv}
Let  $$E\in\{W^{1,1}_0([0,1],\mathbb{R}^{d}),
W^{1,2}_0([0,1],\mathbb{R}^{d})\}.$$ Let  $B\in C^\infty(E,L(E_0,Lie(G))$ can be represented as~(\ref{Bx}). 
Then $B$ is a solution to the system
\begin{equation}
\label{YMequation3}
 \left\{
\begin{aligned}
div^g_L B(\sigma)=0\\
\partial_uB(\sigma)v-\partial_vB(\sigma)u+[B(\sigma)u,B(\sigma)v]=0,\text{if $u,v\in E_0$}
\end{aligned}
\right.
\end{equation}
if and only if there is a connection  $A$ which is a solution to the Yang-Mills equations 
\begin{equation}
\label{YMwo}
g^{\lambda\mu}\nabla_\lambda F_{\mu \nu}=0,
\end{equation}
such that 
$B=B^A$.
\end{theorem}
\begin{proof}
Let  $B$ be a solution to system~(\ref{YMequation3}).
The existence of the $C^\infty$-smooth  connection  $A$ such that $B=B^A$ follows from Theorem~\ref{GrossThm}. 
Then, due to Proposition~\ref{divL},
  the equality $D^2_{tr^g_L}U_{1,0}^A(\sigma)=0$ follows  from the equality $div^g_L B^A(\sigma)=0$.
This implies  $A$ is a solution to~(\ref{YMwo}).
The other side of the theorem is trivial.
\end{proof}

\begin{remark}
The parallel transport can be considered as an infinite-dimensional analogue of the general chiral field.
Recall that if $G$ is a matrix Lie group,  the field  $b\colon \mathbb{R}^d\to G$ is a general chiral field.
Its Dirichet integral  (see Ref.~\cite{DNF}) has a form
\begin{equation}
\label{chiral_field_inf}
\frac 12\int_{\mathbb{R}^d}tr(\partial_\mu b(x)\partial^\mu b^{-1}(x))dx=-\frac12\int_{\mathbb{R}^d}tr(Z_\mu (x)Z^\mu(x))dx,
\end{equation}
where $Z_\mu=b^{-1}(x)\partial_\mu b(x)$.
The equations of motion of the chiral field have the  form
\begin{equation}
\label{chiral_field2}
 \left\{
\begin{aligned}
div Z=0\\
\partial_\mu Z_\nu-\partial_\nu Z_\mu+[Z_\mu,Z_\nu]=0, 
\end{aligned}
\right.
\end{equation}
where $Z=(Z_1,\ldots,Z_d)$. System~(\ref{YMequation3}) is an infinite-dimensional analogue of  system~(\ref{chiral_field2}).

 Conservation laws for the 3-dimensional Yang-Mills theory were considered in Ref.~\cite{AV1981}  (see also Refs.~\cite{Polyakov1979,Polyakov1980}). In particular, it was shown that  $B^A$ is the first non-trivial conserved current if the connection $A$ is a solution to the Yang-Mills equations. The conservation of the current was understood  in the sense that the value of the functional divergence on this current is zero. Thus, Theorem~\ref{thmdiv} is a generalization of this result.
\end{remark}

\section{The Yang-Mills-Higgs equations}

Let $V$ be a finite-dimensional vector space.
For any  $h\in \mathbb{R}^d$ consider  a $V$-valued linear functional   $S_h$ defined in the following way. The domain $dom S_h$ of $S_h$ consists of all $T\in L(E,V)$ 
having the following form 
$$
T(u)=\int_0^1 (u(t),\nu_T(dt))_{\mathbb{R}^d},\,u\in E,
$$
where $\nu_T$ is a $V\times(\mathbb{R}^{d})$-valued Borel measure.
The functional  $S_h$ acts on $T \in dom S_h$ as
$$S_h(T)=(h,\nu_T(\{1\}))_{\mathbb{R}^d}.$$

\begin{remark}
\label{rkf}
Let $E=W^{1,2}_{0}([0,1],\mathbb{R}^d)$.
In terms of the orthonormal basis $\{f_n\}$   of  $W^{1,2}_0([0,1],\mathbb{R})$ (see Remark~\ref{nonclassical})  the operator  $S_h$ can be   defined in the following way. 
For any $h\in R^d$ consider  a linear operator   $S^{\{f_n\}}_h\colon dom S^{\{f_n\}}_h \to V$  defined by
\begin{equation}
\label{dsh}
S_h^{\{f_n\}}(T)=T(hf_0)+\sum_{n=1}^\infty\sqrt{2}(-1)^nT(hf_n),
\end{equation}
where  the domain $dom S^{\{f_n\}}_h$ of $S^{\{f_n\}}_h$ consists of those  $T\in L(E,V)$ for which the right side of~(\ref{dsh}) exists. Then $S^{\{f_n\}}_h$ coincides with $S_h$ on  $dom S_h$.
\end{remark}

\begin{definition}
Choose $X=E$, $Z=Y=V$.  % $\mathcal{H}=W^{2,1}_{0,0}([0,1],\mathbb{R}^d)$, 
The endpoint derivation $D_h$  in the direction  $h\in \mathbb{R}^d$ is the differential operator $D^1_{S_h}$. (The symbol $D_\mu$   denotes $D_{p_\mu}$.)
\end{definition}

In this and the next sections we assume that the Greek indices run from 0 to 3.
We assume the  basis $\{p_0,p_1,p_2,p_3\}$ is given  in the Minkowski space $\mathbb{R}^{1,3}$ and the Minkowski metric $\eta$ is given by $diag\{1,-1,-1,-1\}$ in this basis. We will raise and lower indices using the Minkowski metric.
We assume that the connection is defined as a
  $su(N)$-valued $C^\infty$  smooth 1-form  $A_\mu(x)dx^\mu$ on $\mathbb{R}^{1,3}$. 
Below $$E\in\{W^{1,1}_0([0,1],\mathbb{R}^{1,3}),
W^{1,2}_0([0,1],\mathbb{R}^{1,3})\}.$$

The Yang-Mills-Higgs equations are the following equations on the  $C^\infty$-smooth $su(N)$-valued 1-form $A$ and $C^\infty$-smooth $su(N)$-valued field  $\phi$:
\begin{equation}
\label{higgsequation2}
 \left\{
\begin{aligned}
\nabla^\mu \nabla_\mu\phi(x)-(m^2-l\,tr(\phi^\ast(x)\phi(x)))\phi(x)=0 \\
\nabla^\mu F_{\mu\nu}(x)-[\phi(x),\nabla_\nu\phi(x)]=0,
\end{aligned}
\right.
\end{equation}
where $m,l\geq0$.

\begin{remark}
System of the differential equations~(\ref{higgsequation2}) is the Euler-Lagrange equations for the Lagrangian with  the Lagrangian density:
\begin{multline*}
\mathcal{L}(A_\nu,\partial_\mu A_\nu,\phi,\partial_\mu \phi)=\frac 14tr(F_{\mu\nu}^\ast F^{\mu\nu})+\frac 12tr((\nabla_\mu\phi)^\ast\nabla^\mu\phi)+\\+\frac{m^2}2tr(\phi^\ast\phi)-\frac l4(tr(\phi^\ast\phi))^2.
\end{multline*}
\end{remark}

 In terms of the parallel transport and the L\'evy d'Alembertian solutions to the Yang-Mills-Higgs equations can be described as follows.

  Let the function  $\Phi^{A,\phi}\colon E \to su(N)$  be defined by
$$
\Phi^{A,\phi}(\sigma)=U^A_{0,1}(\sigma)\phi(\sigma(1))U^A_{1,0}(\sigma),\sigma\in E.
$$

\begin{theorem}
\label{YMH1}
A pair $(A,\phi)$ is a  solution to system~(\ref{higgsequation2}) if and only if
for the parallel transport $U^A_{1,0}$ and the function $\Phi^{A,\phi}$
the following relations hold
\begin{equation}
\label{higgsequation3}
 \left\{
\begin{aligned}
\eta^{\mu\nu}D_\mu D_\nu\Phi^{A,\phi}(\sigma)-(m^2-l\,\text{tr}((\Phi^{A,\phi})^\ast(\sigma)\Phi^{A,\phi}(\sigma)))\Phi^{A,\phi}(\sigma)=0 \\
\Box_{L}U^A_{1,0}(\sigma)+U^A_{1,0}(\sigma)\int_0^1[\Phi^{A,\phi}(\sigma^r),D_\nu\Phi^{A,\phi}(\sigma^r)]\dot{\sigma}^\nu(r)dr=0.
\end{aligned}
\right.
\end{equation}
\end{theorem}
\begin{proof}
Since
$$
d_u\Phi^{A,\phi}(\sigma)
=[B^A(\sigma)(u),\Phi^{A,\phi}(\sigma)]+U^A_{0,1}(\sigma)\nabla_\mu\phi(\sigma(1))u^\mu(1)U^A_{1,0}(\sigma),
$$
we have
\begin{equation}
\label{Higgs1}
D_\nu\Phi^{A,\phi}(\sigma)=U^{A}_{0,1}(\sigma)\nabla_\nu\phi(\sigma(1))U^A_{1,0}(\sigma).
\end{equation}
Since
$$
d_uD_\nu\Phi^{A,\phi}(\sigma)
=[B^A(\sigma)(u),D_\nu\Phi^{A,\phi}(\sigma)]+U^A_{0,1}(\sigma)\nabla_\mu\nabla_\nu\phi(\sigma(1))u^\mu(1)U^A_{1,0}(\sigma),
$$
we obtain
\begin{equation}
\label{Higgs2}
D_\mu D_\nu\Phi^{A,\phi}(\sigma)=U^{A}_{0,1}(\sigma)\nabla_\mu\nabla_\nu\phi(\sigma(1))U^A_{1,0}(\sigma).
\end{equation}
Relations~(\ref{Higgs1})
and~(\ref{Higgs2}) and
\begin{multline*}
(m^2-l \text{tr}((\Phi^{A,\phi})^\ast(\sigma)\Phi^{A,\phi}(\sigma)))\Phi^{A,\phi}(\sigma)=\\
=U_{0,1}^{A}(\sigma)(m^2-l\, \text{tr}(\phi^\ast(\sigma(1))\phi(\sigma(1)))\phi(\sigma(1))U^{A}_{1,0}(\sigma)
\end{multline*}
imply that the pair $(A,\phi)$ satisfies the first equation of system~(\ref{higgsequation2}) if and only if the pair  $(U_{1,0},\Phi)$  satisfies the first equation of system~(\ref{higgsequation3}).
Since
\begin{multline*}
\int_0^1 U^A_{1,r}(\sigma)[\phi(\sigma(r)),\nabla_\nu\phi(\sigma(r))]\dot{\sigma}^\nu(r)U^A_{r,0}(\sigma)dr=\\
=U^A_{1,0}(\sigma)\int_0^1[\Phi^{A,\phi}(\sigma^r),D_\nu\Phi^{A,\phi}(\sigma^r)]\dot{\sigma}^\nu(r)dr,
\end{multline*}
the pair  $(A,\phi)$ satisfies the second equation of  system~(\ref{higgsequation2}) if and only if the pair  $(U^A_{1,0},\Phi^{A,\phi})$ satisfies the second equation of  system~(\ref{higgsequation3}). 
\end{proof}

Solutions to the Yang-Mills-Higgs equations can be described in terms of the L\'evy divergence  as follows.

\begin{theorem}
\label{YMH2}
Let  $B\in C^\infty(E,L(E_0,su(N))$  can be represented as~(\ref{Bx}).
Then a pair $(B,\Phi)$, where $\Phi\in C^\infty(E,su(N))$,
is a solution to the system
\begin{equation}
\label{ymhsistem2}
 \left\{
\begin{aligned}
\partial_uB(\sigma)v-\partial_vB(\sigma)u+[B(\sigma)u,B(\sigma)v], \text{if $u,v\in E_0$} \\
div^\eta_LB(\sigma)+\int_0^1[\Phi(\sigma^r),D_\nu\Phi(\sigma^r)]\dot{\sigma}^\nu(r)dr=0\\
\eta^{\mu\nu}D_\mu D_\nu\Phi(\sigma)-(m^2-l \text{tr}(\Phi^\ast(\sigma)\Phi(\sigma)))\Phi(\sigma)=0 \\
\partial_u\Phi(\sigma)+[B(\sigma)u,\Phi(\sigma)]=0, \text{if $u\in E_0$}
\end{aligned}
\right.
\end{equation}
where $m,l\geq0$, if and only if there is  a pair $(A,\phi)$,
  where $A$ is a $su(N)$-valued $C^\infty$-smooth 1-form and $\phi\in C^\infty(\mathbb{R}^{1,3},su(N))$, which is a solution  of the Yang-Mills-Higgs equations~(\ref{higgsequation2}) such that
 $B=B^A$ and $\Phi=\Phi^{A,\phi}$.
%$$
%\Phi(\sigma)=U^{A}_{0,1}(\sigma)\phi(\sigma(1))U^A_{1,0}(\sigma).
%$$
\end{theorem}
\begin{proof}
Let  the pair  $(B,\Phi)$ be a solution to  system~(\ref{ymhsistem2}). Due to  Theorem~\ref{GrossThm} 
the second equation of  the system implies that there is a connection $A$ such that $B=B^A$. Consider the function  $\Phi_1\colon  E\to su(N)$ defined by the formula $\Phi_1(\sigma)=U^A_{1,0}(\sigma)\Phi(\sigma)U^A_{0,1}(\sigma)$. Then  $\Phi_1(\sigma)\in C^\infty(E,su(N))$. Due to the third equation of  system~(\ref{ymhsistem2}), if $u\in E_0$, then
$$
\partial_u\Phi_1(\sigma)=U^A_{1,0}(\sigma)(\partial_u\Phi(\sigma)+[B(\sigma)u,\Phi(\sigma)])U^A_{0,1}(\sigma)=0.
$$
This implies that the field   $\phi\colon \mathbb{R}^{1,3}\to su(N)$ is 
well defined by $\phi(x)=\Phi_1(\sigma_x)$, where $\sigma_x$ is an arbitrary curve from $E$ with $\sigma_x(1)=x$.  It is valid that $\Phi=\Phi^{A,\phi}$.  Since $\Phi$, $U^A_{1,0}$ and $U^A_{0,1}$  are  $C^\infty$-smooth functions,
 $\phi \in C^\infty(\mathbb{R}^{1,3},su(N))$. 
Note that the equality
$$
div^\eta_LB^A(\sigma)+\int_0^1[\Phi^{A,\phi}(\sigma^r),D_\nu\Phi^{A,\phi}(\sigma^r)]\dot{\sigma}^\nu(r)dr=0
$$
holds if and only if the equality
$$
\Box_{L}U^A_{1,0}(\sigma)+U^A_{1,0}(\sigma)\int_0^1[\Phi^{A,\phi}(\sigma^r),D_\nu\Phi^{A,\phi}(\sigma^r)]\dot{\sigma}^\nu(r)dr=0
$$
holds.
Then we can apply Theorem~\ref{YMH1}. It follows  that $(A,\phi)$ is a solution to the Yang-Mills-Higgs equations.

If the pair  $(A,\phi)$ is a solution to the Yang-Mills-Higgs equations, it can be checked by direct calculations that the pair $(B^A,\Phi^{A,\phi})$ is a solution to  system~(\ref{ymhsistem2}). 
\end{proof}

\begin{remark}
Let $(A,\phi)$ and $(A',\phi')$ be such that $B^{A}=B^{A'}$ and $\Phi^{A,\phi}=\Phi^{A',\phi'}$. Then there exists a unique function  $b\in C^\infty(\mathbb{R}^{1,3},SU(N))$ with $b(0)=I_N$ such that $A(x)=b^{-1}(x)A'(x)b(x)+b^{-1}(x)db(x)$ and $\phi(x)=b^{-1}(x)\phi'(x)b(x)$.
\end{remark}

\section{The Yang-Mills-Dirac equations}
Let  $\{g_\alpha\}_{\alpha=1}^4$ be an orthonormal basis in $\mathbb{C}^4$.
  If $\varphi\in \mathbb{C}^N\otimes\mathbb{C}^4$, the symbols $\varphi_\alpha$ ($\alpha=1,\ldots,4$)
  denote such vectors from $\mathbb{C}^N$ that
  $\varphi=\sum_{\alpha=1}^4\varphi_\alpha\otimes g_\alpha$.
  Let $\gamma^\mu$  be the Dirac matrices (we assume
  $\gamma^\mu\in M_4(\mathbb{C})$ and
 $\gamma^\mu\gamma^\nu+\gamma^\nu\gamma^\mu=2\eta^{\mu\nu}I_4$, $\gamma_0^\ast=\gamma_0$, $\gamma_\mu^*=-\gamma_\mu$).
  The symbol $\overline{\varphi}\gamma_\mu\varphi$ denotes the operator acting in $\mathbb{C}^N$ defined by $\overline{\varphi}\gamma_\mu\varphi=\sum_{\alpha=1}^4
((I_N\otimes\gamma_0\gamma_\mu)\varphi)_\alpha\otimes\varphi_\alpha^\ast.$ 
I. e.,
$$
(\sum_{\alpha=1}^4
((I_N\otimes\gamma_0\gamma_\mu)\varphi)_\alpha\otimes\varphi_\alpha^\ast)\xi=
\sum_{\alpha=1}^4(\xi,\varphi_\alpha)((I_N\otimes\gamma_0\gamma_\mu)\varphi)_\alpha,\text{if $\xi\in\mathbb{C}^N$}.
$$
For all $\varphi\in \mathbb{C}^N\otimes\mathbb{C}^4$ it holds that $i\overline{\varphi}\gamma_\mu\varphi \in u(N)$.
The symbol $pr_{su(N)}$ denotes the orthogonal projection in  $u(N)$ on  $su(N)$.

The QCD equations  (the Yang-Mills-Dirac equations)  are the following differential equations on a $C^\infty$-smooth $su(N)$-valued 1-form $A$ and a $C^\infty$-smooth $\mathbb{C}^N\otimes\mathbb{C}^4$-valued function  $\psi$ on $\mathbb{R}^{1,3}$:
\begin{equation}
\label{qcdequation}
 \left\{
\begin{aligned}
(I_N\otimes\gamma^\mu)(\partial_\mu+A_\mu\otimes I_4)\psi+i m\psi=0\\
\nabla^\mu F_{\mu\nu}=-pr_{su(N)}(i(\overline{\psi}\gamma_\nu\psi))
\end{aligned}
\right.
\end{equation}
 \begin{remark}
System of differential equations~(\ref{qcdequation})  is the Euler-Lagrange equations for the Lagrangian with
the Lagrangian density
$$
\mathcal{L}(A_\nu,\partial_\mu A_\nu,\psi,\partial_\mu \psi)=\frac 14 tr(F_{\mu\nu}^\ast F^{\mu\nu})+\overline{\psi}(i(I_N\otimes\gamma^\mu)((\partial_\mu+A_{\mu}\otimes I_4)-m)\psi),
$$
where $\overline{\psi}=\psi^{*}(I_N\otimes\gamma_0)$.
\end{remark}

In terms of the parallel transport and the  L\'evy d'Alembertian the 
solutions to the Yang-Mills-Dirac equations can be described in the  following way.

 The function $\Psi^{A,\phi}\colon E\to\mathbb{C}^N\otimes\mathbb{C}^4$ is defined by the formula
$$
\Psi^{A,\psi}(\sigma)= (U^A_{0,1}(\sigma)\otimes I_4)\psi(\sigma(1)),\sigma\in E.
$$

\begin{theorem}
\label{thmQCD1}
A pair  $(A,\psi)$, where $\psi \in C^\infty(\mathbb{R}^{1,3},\mathbb{C}^N\otimes\mathbb{C}^4)$, is a solution to the Yang-Mills-Dirac equations~(\ref{qcdequation}) if and  only if for the  parallel transport  $U^A_{1,0}$ and the function  $\Psi^{A,\psi}$
the following system holds 
\begin{equation}
\label{qcdequation2}
 \left\{
\begin{aligned}
(I_N\otimes\gamma^\mu) D_\mu\Psi^{A,\psi}+i m\Psi^{A,\psi}=0\\
\Box_{L}U^A_{1,0}(\sigma)=U^A_{1,0}(\sigma)pr_{su(N)}(i\int_0^1\overline{\Psi^{A,\psi}(\sigma^r)}
\gamma_\nu\Psi^{A,\psi}(\sigma^r)\dot{\sigma}^\nu(r)dr).
\end{aligned}
\right.
\end{equation}
\end{theorem}
\begin{proof}
We obtain by direct calculation
$$
D_\mu\Psi^{A,\psi}(\sigma)=(U^{A}_{0,1}(\sigma)\otimes I_4)(\partial_\mu\psi(\sigma(1))+(A_\mu(\sigma(1))\otimes I_4)\psi(\sigma(1))).
$$
This implies the pair $(A,\psi)$ satisfies the first equation of system~(\ref{qcdequation}) if and only if the pair   $(U^A_{1,0},\Psi^{A,\psi})$ satisfies the first equation of  system~(\ref{qcdequation2}). Due to Theorem~\ref{LLYM}  the pair  $(A,\psi)$ is a solution to the second equation of  system~(\ref{qcdequation}) if and only if the following equality holds for all $\sigma\in E$
$$
\Box_{L}U_{1,0}^A(\sigma)=\int_0^1 U^A_{1,r}(\sigma)(pr_{su(N)}(i\overline{\psi(\sigma(r))}\gamma_\nu\psi (\sigma(r))\dot{\sigma}^\nu(r)))U^A_{r,0}(\sigma)dr.
$$
For any  $b\in SU(N)$  and any $\varphi\in \mathbb{C}^N$ the following equality holds (see for example Ref.~\cite{Gross})
$$
pr_{su(N)}(i\overline{(b\otimes I_4)\varphi}\gamma_\mu (b\otimes I_4)\varphi)=pr_{su(N)}(ib\overline{\varphi}\gamma_\mu \varphi b^{-1}).
$$
Hence,
\begin{multline*}
U^A_{1,0}(\sigma)\int_0^1 U^{A}_{0,r}(\sigma)pr_{su(N)}(i\overline{\psi(\sigma(r))}\gamma_\nu\psi (\sigma(r))\dot{\sigma}^\nu(r)U^A_{r,0}(\sigma)dr=\\
=U^A_{1,0}(\sigma)pr_{su(N)}(i\int_0^1\overline{\Psi^{A,\psi}(\sigma^r)}
\gamma_\nu\Psi^{A,\psi}(\sigma^r)\dot{\sigma}^\nu(r)dr).
\end{multline*}
This implies that the pair $(A,\psi)$ satisfies the second equation of system~(\ref{qcdequation}) if and only if the pair   $(U^A_{1,0},\Psi^{A,\psi})$ satisfies the second equation of  system~(\ref{qcdequation2}).
\end{proof}

In terms of the   L\'evy divergence the 
solutions to the Yang-Mills-Dirac equations can be described in the  following way.
\begin{theorem}
\label{thmQCD2}
Let  $B\in C^\infty(E,L(E_0,su(N))$ can be represented as~(\ref{Bx}).
Then a pair $(B,\Psi)$, where  $\Psi\in C^\infty(E,\mathbb{C}^N\otimes\mathbb{C}^4)$,
is a  solution to the system of the equations:
\begin{equation}
\label{qcdequation3}
 \left\{
\begin{aligned}
div^\eta_LB(\sigma)=pr_{su(N)}(i\int_0^1\overline{\Psi(\sigma^r)}
\gamma_\nu\Psi(\sigma^r)\dot{\sigma}^\nu(r)dr)\\
\partial_uB(\sigma)v-\partial_vB(\sigma)u+[B(\sigma)u,B(\sigma)v],\text{if $u,v\in E_0$} \\
\partial_u\Psi(\sigma)+(B(\sigma)u)\Psi(\sigma)=0,\text{if $u\in E_0$}\\
(I_N\otimes\gamma^\mu) D_\mu\Psi+i m\Psi=0,
\end{aligned}
\right.
\end{equation}
if and only if 
there is a pair $(A,\psi)$ which is a solution to the  Yang-Mills-Dirac equations~(\ref{qcdequation}) such that $B=B^A$ and $\Psi=\Psi^{A,\psi}$.  
\end{theorem}
\begin{proof}
Let a pair  $(B,\Psi)$  be a solution to system~(\ref{qcdequation3}). Due to Theorem~\ref{GrossThm}  from the second equation of  system~(\ref{qcdequation3}) it follows that there is a connection  $A$ such that $B=B^A$. Consider the  function $\Psi_1\colon E\to su(N)$ defined by  $\Psi_1(\sigma)=(U^A_{1,0}(\sigma)\otimes I_4)\Psi(\sigma)$. Then $\Psi_1(\sigma)\in C^\infty(E,su(N))$. Due to the third equation of system~(\ref{ymhsistem2}), if $u\in E_0$, then
$$
\partial_u\Psi_1(\sigma)=U^A_{1,0}(\sigma)(\partial_u\Psi(\sigma)+(B^A(\sigma)u)\Psi(\sigma))=0.
$$
Then the field   $\psi\colon \mathbb{R}^{1,3}\to\mathbb{C}^N\otimes\mathbb{C}^4$ is well defined by $\psi(x)=\Psi_1(\sigma_x)$ for arbitrary $\sigma_x\in E$ with $\sigma_x(1)=x$.  It   holds that $\Psi=\Psi^{A,\psi}$. Since $\Psi$ and $U^A_{1,0}$ are  $C^\infty$-smooth functions,
 $\psi \in C^\infty(\mathbb{R}^{1,3},\mathbb{C}^N\otimes\mathbb{C}^4)$. 
Note that  the equality
$$
div^\eta_LB^A(\sigma)=pr_{su(N)}(i\int_0^1\overline{\Psi^{A,\psi}(\sigma^r)}
\gamma_\nu\Psi^{A,\psi}(\sigma^r)\dot{\sigma}^\nu(r)dr).
$$
holds if and only if the following is true
$$
\Box_{L}U^A_{1,0}(\sigma)=U^A_{1,0}(\sigma)pr_{su(N)}(i\int_0^1\overline{\Psi^{A,\psi}(\sigma^r)}
\gamma_\nu\Psi^{A,\psi}(\sigma^r)\dot{\sigma}^\nu(r)dr).
$$
Then we can apply Theorem~\ref{thmQCD1} and obtain that $(A,\psi)$ is a solution  of the Yang-Mills-Dirac equations.

If a pair  $(A,\psi)$ is a solution to the Yang-Mills-Dirac equations, it can be checked by direct calculations that the pair  $(B^A,\Psi^{A,\psi})$ is a solution to  system~(\ref{qcdequation3}). 
 
 \end{proof}

\begin{remark}
In Ref.~\cite{Gross}  systems of infinite-dimensional equations equivalent to the Yang-Mills-Higgs equations  and to the Yang-Mills-Dirac equations were obtained.
The third and the fourth equation of system~(\ref{qcdequation3}) and system~(\ref{ymhsistem2})  were derived in that work. 
\end{remark}

\begin{remark}
Let $(A,\psi)$ and $(A',\psi')$ be such that $B^A=B^{A'}$ and $\Psi^{A,\psi}=\Psi^{A',\psi'}$. Then there exists a unique function $g\in C^\infty(\mathbb{R}^{1,3},SU(N))$ with $b(0)=I_N$ such that $A(x)=b^{-1}(x)A'(x)b(x)+b^{-1}(x)db(x)$ and $\psi(x)=(b^{-1}(x)\otimes I_4)\psi'(x)$.
\end{remark}

\begin{remark}
Due to Remarks~\ref{nonclassical} and~\ref{rkf} all differential operators in systems~(\ref{qcdequation3}) and~(\ref{ymhsistem2}) can be defined by the basis $\{f_n\}$ in $W^{1,2}_0([0,1],\mathbb{R})$.
\end{remark}

\section*{Conclusion}
 
 We provided the classification of the Levy differential operators and obtained the equations contained these operators that are equivalent to the QCD equations and the Yang-Mills equations.
 It would be interesting  to investigate is it possible to find the  stochastic analogue of  Gross's noncommutative
 Poincare lemma.   With the help of such analogue it would be possible to develop the results of the paper Ref.~\cite{Volkov2017}, where the stochastic parallel transport was considered as a general chiral field. Also it would be interesting to investigate 
 the connection between the stochastic L\'evy Laplacians introduced in Ref.~\cite{LV2001} and in Ref.~\cite{Volkov2018}. 
 In the next paper we develop some results of the current paper for the case of the Riemannian and the pseudo-Riemannian manifold.

\section*{Acknowledgments}

The author thanks I.V. Volovich and O.G. Smolyanov for useful discussions.
%This work is supported by the Russian Science Foundation under grant 14-11-00687.

\end{document}